\documentclass[pdflatex,sn-mathphys-num]{sn-jnl}% Math and Physical Sciences Numbered Reference Style
\usepackage{graphicx}%
\usepackage{multirow}%
\usepackage{amsmath,amssymb,amsfonts}%
\usepackage{amsthm}%
\usepackage{mathrsfs}%
\usepackage[title]{appendix}%
\usepackage{algorithm}%
\usepackage{algorithmicx}%
\usepackage{algpseudocode}%
\usepackage{listings}%
\usepackage{faktor}
\usepackage{hyperref}
\usepackage{stmaryrd}
\theoremstyle{thmstyleone}%
\newtheorem{theorem}{Theorem}%  meant for continuous numbers
\newtheorem{proposition}[theorem]{Proposition}% 
\newtheorem{lemma}[theorem]{Lemma}
\theoremstyle{thmstyletwo}%
\newtheorem{example}{Example}%
\newtheorem{remark}{Remark}%
\newtheorem{corollary}{Corollary}%

\theoremstyle{thmstylethree}%

\begin{document}

\title[QEC and EAQEC Codes from Hermitian Sums and Hulls of Cyclic Codes over $\mathbb{F}_2 \times (\mathbb{F}_2+v\mathbb{F}_2)$]{QEC and EAQEC Codes from Hermitian Sums and Hulls of Cyclic Codes over $\mathbb{F}_2 \times (\mathbb{F}_2+v\mathbb{F}_2)$}


\author[1]{\fnm{Rabia} \sur{Zengin}}\email{rabia.zengin@bilgi.edu.tr}

\author*[2]{\fnm{Mehmet Emin} \sur{Köroğlu}}\email{mkoroglu@yildiz.edu.tr}
\equalcont{These authors contributed equally to this work.}

\affil[1]{\orgdiv{Department of Mathematics}, \orgname{Istanbul Bilgi University}, \city{Istanbul}, \postcode{34060}, \state{Eyüpsultan},  \country{Türkiye}}

\affil*[2]{\orgdiv{Department of Mathematics}, \orgname{Yildiz Technical University}, \city{Istanbul}, \postcode{34220}, \state{Esenler}, \country{Türkiye}}


\abstract{In this work, we determine the generator polynomials for the Hermitian hulls and Hermitian sums of cyclic codes defined over the composite ring $\mathbb{F}_2 \times (\mathbb{F}_2 + v\mathbb{F}_2)$, where $v^2 = v$. Based on these structures, we develop quantum error-correcting (QEC) codes by applying the Hermitian dual version of Quantum Construction~X to the obtained Hermitian hulls and sums. Moreover, by employing matrix product code methods on linear complementary dual (LCD) codes defined over the same ring, we derive families of entanglement-assisted quantum error-correcting (EAQEC) codes.}




\maketitle

\section{Introduction}\label{sec1}

Quantum error correction (QEC) emerged alongside growing interest in quantum computers. In some research \cite{benioff1980computer,deutsch1985quantum,feynman2018simulating,simon1997power,yao1993quantum}, it is suggested that quantum computers could outperform classical ones on certain problems. The computational intractability of integer factorization constitutes the foundation of widely implemented public-key cryptographic schemes, such as RSA, which ensure the security of online communications and banking system. QEC techniques are essential for mitigating decoherence in quantum systems; without them, quantum computers would be limited to trivial problem sizes.

Several constructions for these codes have been developed, among which the Calderbank-Shor-Steane (CSS) construction \cite{calderbank1997quantum,steane1996multiple} is particularly significant, as it derives a relationship between classical and quantum codes. In \cite{kai2012new}, quantum MDS codes are obtained from negacyclic codes. In \cite{grassl1999quantum}, quantum BCH codes are constructed. In \cite{ashraf2021new}, some quantum and LCD codes are obtained.

Hsieh et al. \cite{hsieh2007general} introduced a foundational class of quantum codes, termed as entanglement-assisted quantum error-correcting codes (EAQECCs), which coherently integrate the theoretical benefits of entanglement-assisted and operator-based quantum error correction frameworks. In \cite{cao2025entanglement}, EAQECCs are obtained by means of matrix-product codes. In \cite{pereira2021entanglement}, EAQECCs are constructed from algebraic geometry codes. One can look at the references \cite{li2025eaqec,koroglu2019new,pang2021new,sari2021new,chen2018entanglement,li2019entanglement,pereira2022entanglement,sok2022linear}.

In \cite{euclideansum}, QEC and EAQEC codes are obtained by means of Euclidean sums and hulls of cyclic codes over the ring $\mathbb{F}_2\times(\mathbb{F}_2+v\mathbb{F}_2)$. Inspired by \cite{euclideansum}, we will obtain QEC and EAQEC codes via Hermitian sums and hulls of cyclic code over the ring. Specifically, the codes we obtained in Theorem \ref{th12} using Hermitian sums and hulls are derived differently from the code parameters in Theorem 3.5 in \cite{euclideansum}.

This paper is organized as follows. In Section \ref{sec2}, we discuss the ring $\mathbb{F}_2\times(\mathbb{F}_2+v\mathbb{F}_2)$, along with essential definitions and lemmas related to cyclic codes, the Gray map, Euclidean hulls and sums, and Hermitian hulls and sums. In Section \ref{sec3}, we examine the relationship between Hermitian hulls and sums of a linear code under the Gray map. By considering the framework of cyclic codes over $\mathbb{F}_2\times(\mathbb{F}_2+v\mathbb{F}_2)$, we determine the generators of the Hermitian hull and sum of a cyclic code over the ring. Then, we construct QEC codes from these Hermitian hulls and sums in Section \ref{sec4}. Thanks to Hermitian structure, we constructed new QEC codes that different from the codes obtained by Euclidean case. Furthermore, we obtain EAQEC codes through Quantum Construction X and matrix product codes in Section \ref{sec5}. The last section concludes the paper.

\section{Preliminaries}
\label{sec2}
The finite commutative ring
\begin{equation*}
\mathcal{R}:=\mathbb{F}_2\times(\mathbb{F}_2+v\mathbb{F}_2)=\{(u_1,
u_2+vu_3) : u_i\in\mathbb{F}_2, i\in\{1,2,3\}, v^2=v\}
\end{equation*}
is introduced in \cite{F2ring} and then linear and cyclic codes over $\mathcal{R}$ are constructed in \cite{F2ring,cyclicF2,selfdualcodeinring,eaqecinring}.

\noindent For vectors $\mathbf{t}=(t_{1},\dots ,t_{n})$ and $\mathbf{w}=(w_{1},\dots ,w_{n})\in \mathcal{R}=\mathbb{F}_{2}\times (\mathbb{F}_{2}+v\mathbb{F}_{2})$, we define the \emph{Euclidean and} \emph{Hermitian inner
product} by
\begin{equation*}
\left\langle \mathbf{t},\mathbf{w}\right\rangle _{E}=\sum_{i=1}^{n}t_{i}w_{i}
\text{ and }\left\langle \mathbf{t},\mathbf{w}\right\rangle
_{H}=\sum_{i=1}^{n}t_{i}\overline{w_{i}},
\end{equation*}
respectively, where $\overline{w_{i}}$ denotes the conjugate of the element $w_{i}$ under the conjugation map in $\mathbb{F}_{2}+v\mathbb{F}_{2}$. The
conjugation on $\mathbb{F}_{2}+v\mathbb{F}_{2}$ is given by $\overline{b+vc}=b+c+vc$ and it acts component-wise on $\mathcal{R}$ as $\overline{(a,b+vc)}
=(a,b+c+vc).$

\noindent Suppose that $\mathcal{C}$ is a linear code over $\mathcal{R}$.
Then the Euclidean dual and Hermitian dual code of $\mathcal{C}$ are
\begin{equation*}
\mathcal{C}^{\perp }=\{\mathbf{t}\in \mathcal{R}^{n}:\left\langle \mathbf{t},
\mathbf{w}\right\rangle _{E}=(0,0)\text{ for all }\mathbf{w}\in \mathcal{C}%
\},
\end{equation*}
and
\begin{equation*}
\mathcal{C}^{\perp _{H}}=\{\mathbf{t}\in \mathcal{R}^{n}:\left\langle
\mathbf{t},\mathbf{w}\right\rangle _{H}=(0,0)\text{ for all }\mathbf{w}\in
\mathcal{C}\}.
\end{equation*}
\noindent In \cite{F2ring}, the Gray map $\tau $ from $\mathcal{R}$ to $\mathbb{F}_{2}^{3}$ is defined as
\begin{equation*}
\tau :\mathcal{R}\rightarrow \mathbb{F}_{2}^{3},\text{ \ }
(a_{1},a_{2}+va_{3})\mapsto (a_{1}+a_{2},a_{1}+a_{3},a_{1}+a_{2}+a_{3})
\end{equation*}
for all $a_{1},$ $a_{2},$ $a_{3}\in \mathbb{F}_{2}$. The map extends
naturally to $\mathcal{R}^{n}$ as
\begin{equation*}
\tau :\mathcal{R}^{n}\rightarrow \mathbb{F}_{2}^{3n},\text{ \ }(\mathbf{a},
\mathbf{b}+v\mathbf{c})\mapsto (\mathbf{a+b},\mathbf{a+c},\mathbf{a+b}+
\mathbf{c})
\end{equation*}
for all $\mathbf{a},$ $\mathbf{b},$ $\mathbf{c}\in \mathbb{F}_{2}^{n}$.

Suppose that $\mathcal{C}$ is a linear code of length $n$. If $(c_{n-1}, c_0,\dots , c_{n-2})\in \mathcal{C}$ whenever $(c_0, c_1,\dots, c_{n-1})\in \mathcal{C}$, then $\mathcal{C}$ is termed as a cyclic code.

\noindent The polynomial $l(x)\in\mathcal{R}[x]$ can be considered as $l_0+l_1x+l_2x^2+\dots l_{n-1}x^{n-1}$, where $l_j=(a_j,b_j+vc_j)$ for $j\in\{0,1,\dots ,n-1\}$. Let $a(x)=a_0+a_1x+\dots +a_{n-1}x^{n-1}$, $b(x)=b_0+b_1x+\dots +b_{n-1}x^{n-1}$ and $c(x)=c_0+c_1x+\dots +c_{n-1}x^{n-1}$ be in $\mathbb{F}_2[x]$. Then $l(x)$ can be written as $l(x)=(a(x),b(x)+vc(x))$. The multiplication in $\mathcal{R}[x]$ is defined as
\begin{align*}
l(x)t(x) & = (l_1(x),l_2(x)+vl_3(x))(t_1(x),t_2(x)+vt_3(x)) \\
 & = (l_1(x)t_1(x),l_2(x)t_2(x)+v(l_2(x)t_3(x)+l_3(x)t_2(x)+ l_3(x)t_3(x) )),
\end{align*}
where $l(x), t(x)\in \mathcal{R}[x]$.
Let $\mathbf{1}=(1,1)\in\mathcal{R}$. Further, there is a one-to-one correspondence between the polynomial $t(x)=t_0+t_1x+\dots +t_{n-1}x^{n-1}\in \mathcal{R}[x]/\langle \mathbf{1}x^{n}-\mathbf{1} \rangle$ and the element $t=(t_0,t_1,\dots ,t_{n-1})\in\mathcal{R}^n$.

The Hamming weight $wt(x)$ of an element $x$ over $\mathbb{F}_2$ is the number of its nonzero coordinates. In \cite{F2ring}, the Gray weight of $c$ in $\mathcal{R}$ is defined as
\begin{equation*}
wt_G(c) =
\begin{cases}
0, & \text{if } c = (0,0), \\
1, & \text{if } c = (1,1), (1,v), (1,1+v), \\
2, & \text{if } c = (0,1), (0,v), (0,1+v), \\
3, & \text{if } c = (1,0).
\end{cases}
\end{equation*}
Then the Gray weight of $\mathbf{c}=(c_1,c_2,\dots ,c_n)$ in $\mathcal{R}^n$ is $wt_G(\mathbf{c})=\sum_{i=1}^{n}wt_G(c_i)$. The smallest nonzero Gray weight among all codewords of a code is called the Gray weight
of the code. Let $\mathbf{c},\mathbf{c}'\in\mathcal{R}^n$. The Gray
distance between $\mathbf{c}$ and $\mathbf{c}'$ is $d_G(\mathbf{c},\mathbf{c}')=wt_G(\mathbf{c}-\mathbf{c}')$. The smallest nonzero Gray distance between all pairs of distinct
codewords of a code is termed as the
Gray distance of the code.

\begin{lemma} (\cite{F2ring}, Proposition 4, Theorem 3)
\label{graymap} The Gray map $\tau$ defined above satisfies the followings.

\begin{itemize}
\item[(i)] It is an orthogonality preserving map.
\item[(ii)] It is a weight preserving map from $\mathcal{R}^n$ to $\mathbb{F}_{2}^{3n}$ (i.e., from Gray weight to Hamming weight).
\end{itemize}
\end{lemma}

In the following lemma, we recall from Corollary 3 in \cite{F2ring} that the relation between the Euclidean dual of the image of a code over the ring $\mathcal{R}$ under the Gray map $\tau$ and the image of the Hermitian dual of the code under the Gray map $\tau$.

\begin{lemma} (\cite{F2ring}, Corollary 3)
\label{hull} Assume that $\mathcal{C}^{\perp _{H}}$ is the Hermitian dual of
the code $\mathcal{C}$ over $\mathcal{R}$. Then $\tau (\mathcal{C}^{\perp
_{H}})=\tau (\mathcal{C})^{\perp }$.
\end{lemma}

The Euclidean and Hermitian hulls of a linear code $\mathcal{C}$ over the ring $\mathcal{R}$ is defined as $\mathcal{C} \cap \mathcal{C}^{\perp}$ and $\mathcal{C} \cap \mathcal{C}^{\perp_H}$ and denoted by $\mathrm{Hull_E}(\mathcal{C})$ and $\mathrm{Hull_H}(\mathcal{C})$, respectively. If $\mathcal{C}_1$ and $\mathcal{C}_2$ are codes of length $n$ over $\mathcal{R}$, then $\mathcal{C}_1+\mathcal{C}_2= \{ \mathbf{c'} +
\mathbf{c''} : \mathbf{c'}\in \mathcal{C}_1,
\mathbf{c''} \in \mathcal{C}_2 \}$  is the sum of $\mathcal{C}_1$ and $\mathcal{C}_2$. In particular, we denote $\mathcal{C}+\mathcal{C}^{\perp}$ by $\mathrm{Sum_E}(\mathcal{C})$ and $\mathcal{C}+\mathcal{C}^{\perp_H}$ by $\mathrm{Sum_H}(\mathcal{C})$, which are termed as the Euclidean and Hermitian sum of $\mathcal{C}$, respectively.

\section{Hermitian Hulls and Sums of Linear Codes over $\mathbb{F}_2 \times (\mathbb{F}_2+v\mathbb{F}_2)$}
\label{sec3}
In this section, we obtain the Hermitian hulls and sums of cyclic codes over $\mathcal{R}$. First of all, we determine the relation between the Euclidean and Hermitian sums and hulls of linear codes over $\mathcal{R}$.

\begin{lemma}\label{l2}
Suppose that $\mathcal{C}$ is a linear code over $\mathcal{R}$. Then
\begin{itemize}
    \item[(i)] $\tau(\mathrm{Hull_H}(\mathcal{C}))=\mathrm{Hull_E}(\tau(\mathcal{C}))$.
    \item[(ii)] $\mathrm{Sum_E}(\tau(\mathcal{C}))=\tau(\mathrm{Sum_H}(\mathcal{C}))$.
\end{itemize}
\end{lemma}

\begin{proof}

\begin{itemize}
\item[(i)] Let $\mathbf{a}\in \tau (\mathrm{Hull_{H}}(\mathcal{C}))$. Since $%
\tau $ is onto, there exists $\mathbf{x}\in \mathrm{Hull_{H}}(\mathcal{C})$
such that $\tau (\mathbf{x})=\mathbf{a}$. Then $\mathbf{x}\in \mathcal{C}\cap \mathcal{C}^{\perp _{H}}$ and therefore $\mathbf{x}\in \mathcal{C}$ and
$\mathbf{x}\in \mathcal{C}^{\perp _{H}}$. So, $\tau (\mathbf{x})\in \tau (\mathcal{C})$ and $\tau (\mathbf{x})\in \tau (\mathcal{C}^{\perp _{H}})=\tau
(\mathcal{C})^{\perp }$ by Lemma \ref{hull}. Then $\mathbf{a}=\tau (\mathbf{x})\in \tau (\mathcal{C})\cap \tau (\mathcal{C})^{\perp }=\mathrm{Hull_{E}}(\tau (\mathcal{C}))$. Thus, $\tau (\mathrm{Hull_{H}}(\mathcal{C}))\subseteq
\mathrm{Hull_{E}}(\tau (\mathcal{C}))$. On the other hand, assume $\mathbf{a}\in \mathrm{Hull_{E}}(\tau (\mathcal{C}))$. By Lemma \ref{hull}, $\mathbf{a}\in \tau (\mathcal{C})$ and $\mathbf{a}\in \tau (\mathcal{C})^{\perp }=\tau (\mathcal{C}^{\perp _{H}})$. So, there exists some $\mathbf{x}\in \mathcal{C}$
and $\mathbf{y}\in \mathcal{C}^{\perp _{H}}$ such that $\tau (\mathbf{x})=\mathbf{a}$ and $\tau (\mathbf{y})=\mathbf{a}$. Since $\tau $ is injective, $\mathbf{x}=\mathbf{y}$ and therefore $\mathbf{x}\in \mathcal{C}\cap \mathcal{C}^{\perp _{H}}=\mathrm{Hull_{H}}(\mathcal{C})$. Then $\mathbf{a}=\tau
(x)\in \tau (\mathrm{Hull_{H}}(\mathcal{C}))$. Thus, $\tau (\mathrm{Hull_{H}}(\mathcal{C}))\supseteq \mathrm{Hull_{E}}(\tau (\mathcal{C}))$ and hence $\tau (\mathrm{Hull_{H}}(\mathcal{C}))=\mathrm{Hull_{E}}(\tau (\mathcal{C}))$.

\item[(ii)] Let $\mathbf{a}\in \tau (\mathrm{Sum_{H}}(\mathcal{C}))=\tau (\mathcal{C}+\mathcal{C}^{\perp _{H}})$. Then there exists $\mathbf{x}\in \mathcal{C}$ and $\mathbf{y}\in \mathcal{C}^{\perp _{H}}$ such that $\mathbf{a}=\tau (\mathbf{x}+\mathbf{y})$. Since $\tau $ is linear, $\mathbf{a}=\tau (\mathbf{x})+\tau (\mathbf{y})$. Also, $\tau (\mathbf{x})\in \tau (\mathcal{C})$ and $\tau (\mathbf{y})\in \tau (\mathcal{C}^{\perp _{H}})=\tau (\mathcal{C})^{\perp }$. Then $\mathbf{a}\in \tau (\mathcal{C})+\tau (\mathcal{C})^{\perp }=\mathrm{Sum_{E}}(\tau (\mathcal{C}))$. So, $\tau (\mathrm{Sum_{H}}(\mathcal{C}))\subseteq \mathrm{Sum_{E}}(\tau (\mathcal{C}))$. Now, let $\mathbf{b}\in \mathrm{Sum_{E}}(\mathcal{C})$. Then $\mathbf{b}\in \tau (\mathcal{C})+\tau (\mathcal{C}^{\perp _{H}})$ and therefore $\mathbf{b}=\tau (\mathbf{u})+\tau (\mathbf{v})$ for some $\mathbf{u}\in \mathcal{C}$ and $\mathbf{v}\in \mathcal{C}^{\perp _{H}}$. So, $\mathbf{b}=\tau (\mathbf{u}+\mathbf{v})\in \tau (\mathrm{Sum_{H}}(\mathcal{C}))$. So, $\tau (\mathrm{Sum_{H}}(\mathcal{C}))\supseteq \mathrm{Sum_{E}}(\tau (\mathcal{C}))$. Thus, $\tau (\mathrm{Sum_{H}}(\mathcal{C}))=\mathrm{Sum_{E}}(\tau (\mathcal{C}))$.
\end{itemize}
\end{proof}

\noindent We denote the monic reciprocal polynomial of $h(x)=t_0+t_1x+\dots +t_nx^n$ by $\Bar{h}(x)$, i.e., $\Bar{h}(x)=t_n^{-1}x^{n}h(\frac{1}{x})$.

\begin{theorem}\cite{cyclicF2}\label{7.1}
Let $\mathcal{C}$ be a linear code over $\mathcal{R}$. Then
\begin{itemize}
\item[(i)] $\mathcal{C}=\langle \nu(x)\rangle$, where $\nu(x)=(r_1(x),(1+v)r_2(x)+r_3(x))$ and $r_i(x)\mid (x^n-1)$ for all $i\in \{1,2,3\}$. Also, the polynomial $\nu(x)$ is unique and $|\mathcal{C}|=2^{3n-\mathrm{deg}(r_1(x))-\mathrm{deg}(r_2(x))-\mathrm{deg}(r_3(x))}$.
\item[(ii)] $\mathcal{C}^{\perp_H}=\langle \Bar{h}(x)\rangle$, where $\Bar{h}(x)=(\Bar{h}_1(x),v\Bar{h}_2(x)+(1+v)\Bar{h}_3(x))$ and $x^n-1=r_i(x)h_i(x)$ for all $i\in \{1,2,3\}$. Further, $|\mathcal{C}^{\perp_H}|=2^{\mathrm{deg}(r_1(x))+\mathrm{deg}(r_2(x))+\mathrm{deg}(r_3(x))}$.
\end{itemize}
\end{theorem}

\begin{remark}\cite{F2ring}\label{remark7.1}
A linear code $\mathcal{C}$ over $\mathcal{R}$ of length $n$ is permutation equivalent to direct product of $\mathcal{C}_1 $ and $ \mathcal{C}_2$, where $\mathcal{C}_1$ is a binary linear code of length $n$, $\mathcal{C}_2$ is a linear code over $\mathcal{R}$ of length $n$, which will be denoted by $\mathcal{C}= (\mathcal{C}_1,\mathcal{C}_2)$.
\end{remark}

\begin{theorem}\cite{cyclicinring}\label{th7.2}
Suppose that $\Theta$ is the Gray map from $\mathbb{F}_2+v\mathbb{F}_2$ to $\mathbb{F}_2^2$ defined by $\Theta(x+vy)=(x,x+y)$. Let $\mathcal{C}$ be a linear code of length $n$ over $\mathbb{F}_2+v\mathbb{F}_2$.
Then $\Theta(\mathcal{C})=\mathcal{C}_1 \otimes \mathcal{C}_2:=\{(x,y):x\in\mathcal{C}_1,y\in\mathcal{C}_2)\}$ and $|\mathcal{C}|=|\mathcal{C}_1||\mathcal{C}_2|$, where $\mathcal{C}_1=\{x\in\mathbb{F}_2^n : x+vy \in\mathcal{C}\}$ and $\mathcal{C}_2=\{x+y\in\mathbb{F}_2^n : x+vy \in\mathcal{C}\}$. Also, $\Theta(\mathcal{C})$ is linear.
\end{theorem}

\noindent By utilizing Theorem \ref{th7.2}, we get the following proposition.

\begin{proposition}\label{prop1}
Let $d$ and $d_G$ be the minimum Hamming and Gray distances of a linear code $\mathcal{C}$ over $\mathbb{F}_2+v\mathbb{F}_2$, respectively. Then $d=d_G=\mathrm{min}\{d(\mathcal{C}_1),d(\mathcal{C}_2)\}$, where $d(\mathcal{C}_i)$ denotes the
minimum distance of binary codes $\mathcal{C}_i$ for $i=1,2$.
\end{proposition}

\begin{proof}
Since $\Theta$ is a distance preserving map, it follows that $d_G(\mathcal{C})=d(\Theta(\mathcal{C}))=d(\mathcal{C}_1 \otimes \mathcal{C}_2)=\mathrm{min}\{d(\mathcal{C}_1),d(\mathcal{C}_2)\}$. Thus, $d=d_G$.
\end{proof}

\begin{proposition}\label{prop2}
Let $\mathcal{C}=\langle \nu(x)\rangle$, where $\nu(x)=(r_1(x),(1+v)r_2(x)+vr_3(x))$ and $x^n-1=r_i(x)h_i(x)$ for all $i\in \{1,2,3\}$, be a cyclic code over $\mathcal{R}$. If $\mathcal{C}^{\perp_H}=\langle (\Bar{h}_1(x),v\Bar{h}_2(x)+(1+v)\Bar{h}_3(x))\rangle$, then $\tau(\mathcal{C})$ is a binary $[3n,k,d]$ code, where $k=3n-\mathrm{deg}(r_1(x))-\mathrm{deg}(r_2(x))-\mathrm{deg}(r_3(x))$ and $d=\mathrm{min}\{d(\langle r_1(x)\rangle),d(\langle r_2(x)\rangle),d(\langle r_3(x)\rangle)\}$.
\end{proposition}

\begin{proof}
The proof is the same as that of Lemma 2.6 in \cite{euclideansum}.
\end{proof}
\noindent In the next theorem, we present the generator polynomials and the corresponding dimensions of the Hermitian hulls and Hermitian sums associated with cyclic codes over the ring~$\mathcal{R}$.

\begin{theorem}
\label{th3} Let $\mathcal{C}=\langle \nu(x)\rangle$, where $\nu(x)=(r_1(x),(1+v)r_2(x)+r_3(x))$ and $\mathcal{C}^{\perp_H}=\langle \Bar{h}(x)\rangle$, where $\Bar{h}(x)=(\Bar{h}_1(x),v \Bar{h}_2(x)+(1+v)\Bar{h}_3(x))$ and $x^n-1=r_i(x)h_i(x)$ for all $i\in \{1,2,3\}$. Then
\begin{itemize}
\item[(i)] $\mathrm{Hull_{H}}(\mathcal{C})=\langle (\mathrm{lcm}(r_{1}(x),\Bar{h}_{1}(x)),(1+v)\mathrm{lcm}(r_{2}(x),\Bar{h}_{3}(x))+v\mathrm{lcm}(r_{3}(x),\Bar{h}_{2}(x)))\rangle $. Also,
\begin{equation*}
|\mathrm{Hull_{H}}(\mathcal{C})|=2^{3n-\mathrm{\deg }(\mathrm{lcm}(r_{1}(x),\Bar{h}_{1}(x)))-\mathrm{\deg }(\mathrm{lcm}(r_{2}(x),\Bar{h}_{3}(x)))-\mathrm{\deg }(\mathrm{lcm}(r_{3}(x),\Bar{h}_{2}(x)))}.
\end{equation*}

\item[(ii)] $\mathrm{Sum_{H}}(\mathcal{C})=\langle (\mathrm{\gcd }(r_{1}(x),\Bar{h}_{1}(x)),(1+v)\mathrm{\gcd }(r_{2}(x),\Bar{h}_{3}(x))+v\mathrm{\gcd }(r_{3}(x),\Bar{h}_{2}(x)))\rangle $. Also,
\begin{equation*}
|\mathrm{Sum_{H}}(\mathcal{C})|=2^{3n-\mathrm{\deg }(\mathrm{\gcd }(r_{1}(x),\Bar{h}_{1}(x)))-\mathrm{\deg }(\mathrm{\gcd }(r_{2}(x),\Bar{h}_{3}(x)))-\mathrm{\deg }(\mathrm{\gcd }(r_{3}(x),\Bar{h}_{2}(x)))}.
\end{equation*}
\end{itemize}
\end{theorem}

\begin{proof}
\begin{itemize}
\item[(i)] Assume that
\begin{equation*}
\mathcal{A}=\langle (c_{1}(x),(1+v)c_{2}(x)+vc_{3}(x))\rangle ,
\end{equation*}
where $c_{1}(x)=\mathrm{lcm}(r_{1}(x),\Bar{h}_{1}(x))$, $c_{2}(x)=\mathrm{lcm}(r_{2}(x),\Bar{h}_{3}(x))$ and $c_{3}(x)=\mathrm{lcm}(r_{3}(x),\Bar{h}_{2}(x)))$, is a cyclic code of length $n$ over $\mathcal{R}$. Then  there exist $a_{i}(x),$ $b_{i}(x)\in \mathbb{F}_{2}[x]$ and $i\in \{1,2,3\}$ such that
\begin{align*}
c_{1}(x)& =r_{1}(x)a_{1}(x)=\Bar{h}_{1}(x)b_{1}(x), \\
c_{2}(x)& =r_{2}(x)a_{2}(x)=\Bar{h}_{3}(x)b_{2}(x), \\
c_{3}(x)& =r_{3}(x)a_{3}(x)=\Bar{h}_{2}(x)b_{3}(x).
\end{align*}
So, $\mathcal{A}\subseteq \mathcal{C}$ and $\mathcal{A}\subseteq \mathcal{C}^{\perp }$ since
\begin{equation*}
(c_{1}(x),(1+v)c_{2}(x)+vc_{3}(x))=(a_{1}(x),(1+v)a_{2}(x)+va_{3}(x))\times
(r_{1}(x),(1+v)r_{2}(x)+vr_{3}(x)),
\end{equation*}
\begin{equation*}
(c_{1}(x),(1+v)c_{2}(x)+vc_{3}(x))=(b_{1}(x),(1+v)b_{2}(x)+vb_{3}(x))\times (\Bar{h}_{1}(x),v\Bar{h}_{2}(x)+(1+v)\Bar{h}_{3}(x)).
\end{equation*}
Therefore, $\mathcal{A}\subseteq \mathrm{Hull_{H}}(\mathcal{C})$. Clearly, $\mathrm{Hull_{H}}(\mathcal{C})$ is a cyclic code of length $n$ over $%
\mathcal{R}$. This implies that $\mathrm{Hull_{H}}(\mathcal{C})=\langle
(e_{1}(x),(1+v)e_{2}(x)+ve_{3}(x))\rangle $ for some $e_{1}(x),e_{2}(x),e_{3}(x)\in \mathbb{F}_{2}[x]$ such that $e_{i}(x)\mid
(x^{n}-1)$ for $i\in \{1,2,3\}$. Due to $\mathrm{Hull_{H}}(\mathcal{C}%
)\subset \mathcal{C}$,
\begin{align*}
(e_{1}(x),(1+v)e_{2}(x)+ve_{3}(x))&
=(s_{1}(x),(1+v)s_{2}(x)+vs_{3}(x))\times (r_{1}(x),(1+v)r_{2}(x)+vr_{3}(x))
\\
& =(s_{1}(x)r_{1}(x),(1+v)s_{2}(x)r_{2}(x)+vs_{3}(x)r_{3}(x))
\end{align*}
for some $s_{1}(x),$ $s_{2}(x),$ $s_{3}(x)\in \mathbb{F}_{2}[x]$. Then $%
r_{i}(x)\mid e_{i}(x)$ for $i\in \{1,2,3\}$. Similarly, $\mathrm{Hull_{H}}(%
\mathcal{C})\subset \mathcal{C}^{\perp }$ implies $\Bar{h}_{1}(x)\mid
e_{1}(x)$, $\Bar{h}_{2}(x)\mid e_{3}(x)$ and $\Bar{h}_{3}(x)\mid e_{2}(x)$.
Thus, $c_{i}(x)\mid e_{i}(x)$ and therefore there exists $p_{i}(x)\in
\mathbb{F}_{2}[x]$ such that $e_{i}(x)=c_{i}(x)p_{i}(x)$ for $i\in \{1,2,3\}$. Then
\begin{equation*}
(e_{1}(x),(1+v)e_{2}(x)+ve_{3}(x))=(p_{1}(x),(1+v)p_{2}(x)+vp_{3}(x))(c_{1}(x),(1+v)c_{2}(x)+vc_{3}(x)).
\end{equation*}
This means $\mathrm{Hull_{H}}(\mathcal{C})\subseteq \mathcal{A}$. Hence,
\begin{equation*}
\mathrm{Hull_{H}}(\mathcal{C})=\langle (\mathrm{lcm}(r_{1}(x),\Bar{h}_{1}(x)),(1+v) \mathrm{lcm}(r_{2}(x),\Bar{h}_{3}(x))+v\mathrm{\mathrm{lcm}}(r_{3}(x),\Bar{h}_{2}(x)))\rangle .
\end{equation*}
By Theorem \ref{7.1} part (i), we get
\begin{equation*}
|\mathrm{Hull_{H}}(\mathcal{C})|=2^{3n-\mathrm{\deg }(\mathrm{lcm}(r_{1}(x),\Bar{h}_{1}(x)))-\mathrm{\deg }(\mathrm{lcm}(r_{2}(x),\Bar{h}_{3}(x)))-\mathrm{\deg }(\mathrm{lcm}(r_{3}(x),\Bar{h}_{2}(x)))}.
\end{equation*}

\item[(ii)] The proof is similar to part (i).
\end{itemize}
\end{proof}

\section{QEC Codes from Hermitian Sums of Cyclic Codes over $\mathbb{F}_2\times(\mathbb{F}_2+v\mathbb{F}_2)$}
\label{sec4}
In this section, we obtain QEC codes by virtue of Hermitian hulls and sums of cyclic codes over $\mathcal{R}$.

\noindent An explicit method to get QEC codes from classical linear codes which is called quantum construction X of the Hermitian dual is given as follows.

\begin{theorem}[Quantum Construction X]
\cite{quantumxhermitian}\label{quantumxhermitian} Assume that $\mathcal{C}$ is an $[n,k]$-linear code over $\mathbb{F}_{p^2}$. Then there exists a quantum code over $\mathbb{F}_p$
with parameters $\llbracket n+e,2k-n,d\rrbracket $, where $d \geq \mathrm{min} \{d(\mathcal{C}%
),d(\mathrm{Sum_H}(\mathcal{C}))+1\}$ and $e = n-k-\mathrm{dim}(\mathrm{Hull_H}(\mathcal{C}))$.
\end{theorem}

\begin{theorem}
\label{th12} Let $\mathcal{C}=\langle
(r_{1}(x),(1+v)r_{2}(x)+r_{3}(x))\rangle $ and $x^{n}-1=r_{i}(x)h_{i}(x)$
for all $i\in \{1,2,3\}$. If
\begin{equation*}
\eta =\mathrm{deg}(\mathrm{lcm}(r_1(x),\Bar{h}_1(x)))+\mathrm{deg}(\mathrm{lcm}(r_2(x),\Bar{h}_3(x)))+\mathrm{deg}(\mathrm{lcm}(r_3(x),\Bar{h}_2(x)))
\end{equation*}
and
\begin{equation*}
\mu ={{\overset{3}{\underset{i=1}{{\sum }}}}}\mathrm{\deg }(r_{i}(x)),
\end{equation*}
then there exists a binary $\llbracket\eta +\mu -2n,n-2\mu ,d\geq \mathrm{\min }\{d(\tau (\mathcal{C})),d(\tau (\mathrm{Sum_{H}}(\mathcal{C})))+1\}\rrbracket$ QEC code, where
\begin{equation*}
d(\tau (\mathcal{C}))=\mathrm{\min }\{d(\langle r_{1}(x)\rangle ),d(\langle
r_{2}(x)\rangle ),d(\langle r_{3}(x)\rangle )\}
\end{equation*}
and
\begin{equation*}
d(\tau (\mathrm{Sum_{H}}(\mathcal{C})))=\mathrm{\min }\{d(\langle \mathrm{\gcd }(r_{1}(x),\Bar{h}_{1}(x))\rangle ),d(\langle \mathrm{\gcd }(r_{2}(x),\Bar{h}_{3}(x))\rangle ),d(\langle \mathrm{\gcd }(r_{3}(x),\Bar{h}_{2}(x))\rangle )\}.
\end{equation*}
\end{theorem}

\begin{proof}
By Theorem \ref{th3} part $(i)$,
\begin{equation*}
\mathrm{dim}(\mathrm{Hull_H}(\mathcal{C})))=3n-\mathrm{deg}(\mathrm{lcm}(r_1(x),\Bar{h}_1(x)))-\mathrm{deg}(\mathrm{lcm}(r_2(x),\Bar{h}_3(x)))-\mathrm{deg}(\mathrm{lcm}(r_3(x),\Bar{h}_2(x))).
\end{equation*}
From Theorem \ref{th3} part $(ii)$ and Lemma \ref{l2}, $d(\mathrm{Sum_{H}}(\tau (\mathcal{C})))=d(\tau (\mathrm{Sum_{H}}(\mathcal{C})))=\mathrm{\min } \{d(\langle \mathrm{\gcd }(r_{1}(x),\Bar{h}_{1}(x))\rangle ),d(\langle \mathrm{\gcd }(r_{2}(x),\Bar{h}_{3}(x))\rangle ),d(\langle \mathrm{\gcd }(r_{3}(x),\Bar{h}_{2}(x))\rangle )\}$. Moreover, $d(\tau (\mathcal{C}))=d(\mathcal{C})=\mathrm{\min }\{d(\langle r_{1}(x)\rangle ),d(\langle r_{2}(x)\rangle ),d(\langle r_{3}(x)\rangle )\}$. Hence, there exists a binary $\llbracket\eta +\mu -2n,n-2\mu ,d\geq \mathrm{\min }\{d(\tau (\mathcal{C})),d(\tau (\mathrm{Sum_{H}}(\mathcal{C})))+1\}\rrbracket$ QEC code by Theorem \ref{quantumxhermitian}.
\end{proof}

\begin{example}
The polynomial $x^{22}-1\in \mathbb{F}_{2}[x]$ is factorized as
\begin{equation*}
x^{22}-1=(x+1)^{2}\left(
x^{10}+x^{9}+x^{8}+x^{7}+x^{6}+x^{5}+x^{4}+x^{3}+x^{2}+x+1\right) ^{2}
\end{equation*}
Let $r_{1}(x)=x+1$ and $r_{2}(x)=r_{3}(x)=(x+1)^{2}$ and $\mathcal{C}=\langle (x+1)^{2},(1+v)(x+1)^{2}+(x+1)^{2})\rangle $. Then
\begin{eqnarray*}
\mathcal{C}^{\perp _{H}} &=&\langle ((x+1)\left(
x^{10}+x^{9}+x^{8}+x^{7}+x^{6}+x^{5}+x^{4}+x^{3}+x^{2}+x+1\right) ^{2}, \\
&&v\left( x^{10}+x^{9}+x^{8}+x^{7}+x^{6}+x^{5}+x^{4}+x^{3}+x^{2}+x+1\right)
\\
&&+(1+v)\left(
x^{10}+x^{9}+x^{8}+x^{7}+x^{6}+x^{5}+x^{4}+x^{3}+x^{2}+x+1\right) )\rangle .
\end{eqnarray*}
So,
\begin{equation*}
\mathrm{Sum_{H}}(\mathcal{C})=\langle ((x+1),(1+v)(1)+v(1))\rangle
\end{equation*}
and
\begin{eqnarray*}
\mathrm{Hull_{H}}(\mathcal{C}) &=&\langle ((x+1)\left(
x^{10}+x^{9}+x^{8}+x^{7}+x^{6}+x^{5}+x^{4}+x^{3}+x^{2}+x+1\right) ^{2}, \\
&&(1+v)(x^{22}-1)+v(x^{22}-1))\rangle
\end{eqnarray*}
Also,
\begin{equation*}
d(\tau (\mathcal{C}))=\mathrm{\min }\{d(\langle r_{1}(x)\rangle ),d(\langle
r_{2}(x)\rangle ),d(\langle r_{3}(x)\rangle )\} \\
=\mathrm{\min }\{2,2,2\} \\
=2
\end{equation*}
and
\begin{eqnarray*}
d(\tau (\mathrm{Sum_{H}}(\mathcal{C}))) &=&\mathrm{\min }\{d(\langle \mathrm{\gcd }(r_{1}(x),\Bar{h}_{1}(x))\rangle ), \\
&&d(\langle \mathrm{\gcd }(r_{2}(x),\Bar{h}_{3}(x))\rangle ),d(\langle
\mathrm{\gcd }(r_{3}(x),\Bar{h}_{2}(x))\rangle )\} \\
&=&\mathrm{\min }\{2,1,1\} \\
&=&1
\end{eqnarray*}
Thus, there exists a binary $\llbracket26,12,\geq 2\rrbracket$ QEC code by
Theorem \ref{th12}.
\end{example}
In Table \ref{TableTh10}, we tabulated some parameters of QECCs based on Theorem \ref{th12}. Most of the obtained codes are among the best known QECCs  according to \texttt{https://www.codetables.de/}( \cite{Grassl:codetables}) and some of them are optimal with respect to quantum singleton bound. All the computations are done using \verb"MAGMA" computer program \cite{bosma1997magma}.

\begin{table}[!ht]
\caption{Some binary QEC Codes from Theorem \ref{th12}}\label{TableTh10}
\begin{tabular}{lccc}
\toprule
$n$ & generators of $\mathcal{C}_{1}$ & generators of $\mathcal{C}_{2}$ & QECC
parameters \\
\midrule
$4$ & $1$ & $(1+v)(1) + v(x + 1)$ & $\llbracket 4,2,d\ge 2\rrbracket$   \\
$15$ & $1$ & $(1+v)(1)+v(x^{4}+x+1)$ & $\llbracket15,7,d\geq 3\rrbracket$ \\
$15$ & $1$ & $(1+v)(1)+v(x^{5}+x^{3}+x+1)$ & $\llbracket15,5,d\geq 4%
\rrbracket$ \\
$15$ & $1$ & $(1+v)(1)+v(x^{6}+x^{4}+x^{3}+x^{2}+1)$ & $\llbracket15,3,d\geq
4\rrbracket$ \\
$15$ & $x^{5}+x^{3}+x+1$ & $(1+v)(1)+v(1)$ & $\llbracket16,5,d\geq 4%
\rrbracket$ \\
$26$ & $1$ & $(1+v)(1)+v(1+x)$ & $\llbracket26,24,d\geq 2\rrbracket$ \\
$26$ & $1$ & $(1+v)(1)+v(x^{2}+1)$ & $\llbracket26,22,d\geq 2\rrbracket$ \\
$26$ & $1+x$ & $(1+v)(1)+v(x^{2}+1)$ & $\llbracket26,20,d\geq 2\rrbracket$
\\
$28$ & $1$ & $(1+v)(1)+v(1+x)$ & $\llbracket28,26,d\geq 2\rrbracket$
\\
$28$ & $1$ & $(1+v)(1)+v(x^{2}+1)$ & $\llbracket28,24,d\geq 2\rrbracket$ \\
$28$ & $1$ & $(1+v)(1)+v(x^{3}+x+1)$ & $\llbracket28,22,d\geq 2\rrbracket$
\\
$28$ & $1$ & $(1+v)(1)+v(x^{6}+x^{3}+x+1)$ & $\llbracket28,16,d\geq 4%
\rrbracket$ \\
$28$ & $x^{6}+x^{3}+x+1$ & $(1+v)(1)+v(1)$ & $\llbracket30,16,d\geq 3%
\rrbracket$ \\
$31$ & $1$ & $(1+v)(1)+v(1+x)$ & $\llbracket31,29,d= 1\rrbracket$ \\
$31$ & $1$ & $(1+v)(1)+v(x^{5}+x^{2}+1)$ & $\llbracket31,21,d\geq 3\rrbracket
$ \\
$31$ & $1$ & $(1+v)(1)+v(x^{6}+x^{2}+x+1)$ & $\llbracket31,19,d\geq 4%
\rrbracket$ \\
$31$ & $1+x$ & $(1+v)(1)+v(1+x)$ & $\llbracket32,27,d\geq 2\rrbracket$ \\
$31$ & $x^{6}+x^{2}+x+1$ & $(1+v)(1)+v(1)$ & $\llbracket32,19,d\geq 4%
\rrbracket$ \\
$31$ & $x^{11}+x^{6}+x^{5}+x^{2}+x+1$ & $(1+v)(1)+v(1)$ & $\llbracket%
32,9,d\geq 6\rrbracket$ \\
$36$ & $1$ & $(1+v)(1)+v(1+x)$ & $\llbracket36,34,d\geq 2\rrbracket$ \\
$36$ & $1$ & $(1+v)(1)+v(x^{2}+1)$ & $\llbracket36,32,d\geq 2\rrbracket$ \\
$36$ & $1$ & $(1+v)(1)+v(x^{3}+1)$ & $\llbracket36,30,d\geq 2\rrbracket$ \\
$36$ & $1$ & $(1+v)(1)+v(x^{9}+x^{8}+x^{7}+x^{5}+x^{4}+x^{2}+x+1)$ & $%
\llbracket36,18,d\geq 4\rrbracket$ \\
$48$ & $1$ & $(1+v)(1)+v(1+x)$ & $\llbracket48,46,d\geq 2\rrbracket$ \\
$48$ & $1$ & $(1+v)(1)+v(x^{2}+1)$ & $\llbracket48,44,d\geq 2\rrbracket$ \\
$48$ & $1$ & $(1+v)(1)+v(x^{3}+1)$ & $\llbracket48,42,d\geq 2\rrbracket$ \\
\botrule
\end{tabular}
\end{table}

\section{EAQEC codes from Hermitian Sums of Cyclic Codes over $\mathbb{F}_2 \times (\mathbb{F}_2+v\mathbb{F}_2)$}
\label{sec5}
In this section, we get EAQEC codes by using the following proposition.

\begin{proposition}\label{eaqec}\cite{eaqec}
Assume that $\mathcal{C}$ is a $q$-ary $[n,k,d]$-code. If $\mathcal{C}^{\perp}$ is its Euclidean dual code with parameters $[n,n - k,d^{\perp}]_q$, then there exist $q$-ary $\llbracket n,k -\mathrm{dim}(\mathrm{Hull}_{E}(\mathcal{C})),d;n-k-\mathrm{dim}(\mathrm{Hull}_{E}(\mathcal{C}))\rrbracket $ and $\llbracket n,n-k-\mathrm{dim}(\mathrm{Hull}_{E}(\mathcal{C})),d^\perp;k-\mathrm{dim}(\mathrm{Hull}_{E}(\mathcal{C}))\rrbracket $
 EAQEC codes.
\end{proposition}

\noindent Let $\mathcal{C}$ be a linear code. If $\mathrm{Hull}_{E}(\mathcal{C})=\{\mathbf{0}\}$, then $\mathcal{C}$ is named an LCD code. Similarly, it is termed as Hermitian LCD code when $\mathrm{Hull_H}(\mathcal{C})=\{\mathbf{0}\}$.

\begin{lemma}\cite{lcd}
Suppose that $\mathcal{C}$ is an $[n,k]$-code. If $G_{k\times n}$ is generator matrix of $\mathcal{C}$, then $\mathcal{C}$ is Hermitian LCD code if and only if $\mathrm{rk}(G\Bar{G}^T)=k$, where $\mathrm{rk}$ and $\Bar{G}^T$ stand for the rank of a matrix and the transpose of conjugate of $G$, respectively.
\end{lemma}

\begin{corollary}\label{corollary1}
Let $\mathcal{C}$ be a linear code of length $n$ and dimension $k$ over the ring~$\mathcal{R}$.
Then $\mathcal{C}$ is a Hermitian LCD code precisely when its Gray image $\tau(\mathcal{C})$ forms a binary LCD code with parameters~$[3n,\,k]$.
\end{corollary}

\begin{proof}
It is clear by Lemma \ref{l2}.
\end{proof}

\noindent A polynomial $s(x)$ of degree $t$ is termed as self-reciprocal if $s(x)=x^ts(\frac{1}{x})$.

\begin{proposition}\label{prop14}
Let $\mathcal{C} = \langle (r_1(x), (1+v)r_2(x) + r_3(x)) \rangle$ be a cyclic code of length~$n$ over the ring~$\mathcal{R}$,
where $x^n - 1 = r_i(x)h_i(x)$ in~$\mathbb{F}_2[x]$ and each $r_i(x)$ is self-reciprocal for $i = 1,2,3$.
If every monic irreducible factor of $r_i(x)$ occurs with the same multiplicity in both $r_i(x)$ and $x^n - 1$ for all $i \in \{1,2,3\}$,
then $\mathcal{C}$ is a Hermitian LCD code.
\end{proposition}

\begin{proof}
Similar with the proof of Proposition 4.6 in \cite{euclideansum}.
\end{proof}

\begin{theorem}\label{last_theorem}
Let $\mathcal{C}_{1}=\langle (r_{1}(x),(1+v)r_{2}(x)+vr_{3}(x))\rangle $ and
$\mathcal{C}_{2}=\langle (s_{1}(x),(1+v)s_{2}(x)+vs_{3}(x))\rangle $ be two
cyclic codes of length $n$ over the ring $\mathcal{R}$. Assume that the
polynomials $r_{i}(x)$ and $s_{i}(x)$ are self-reciprocal, and that each
monic irreducible factor of $r_{i}(x)$ and $s_{i}(x)$ appears with the same
multiplicity in both $r_{i}(x),s_{i}(x)$ and $x^{n}-1$ for all $i\in
\{1,2,3\}$. Define
\begin{equation*}
L=[\,\tau (\mathcal{C}_{1}),\,\tau (\mathcal{C}_{2})\,]B,
\end{equation*}
where
\begin{equation*}
B=
\begin{bmatrix}
1 & 1 & 0 & 1 \\
0 & 1 & 1 & 1
\end{bmatrix}
.
\end{equation*}
Then $L$ is a binary LCD code with parameters $[12n,\,k,\,d\geq \min
\{3d_{1},\,2d_{2}\}]_{2}$, where
\begin{equation*}
k=6n-\sum_{i=1}^{3}\deg (r_{i}(x))-\sum_{j=1}^{3}\deg (s_{j}(x)),
\end{equation*}
\begin{equation*}
d_{1}=\min \{d(\langle r_{1}(x)\rangle ),\,d(\langle r_{2}(x)\rangle
),\,d(\langle r_{3}(x)\rangle )\},
\end{equation*}
and
\begin{equation*}
d_{2}=\min \{d(\langle s_{1}(x)\rangle ),\,d(\langle s_{2}(x)\rangle
),\,d(\langle s_{3}(x)\rangle )\}.
\end{equation*}
Moreover, there exist binary EAQEC codes with parameters $\llbracket 12n,\,k,\,d\geq \min \{3d_{1},\,2d_{2}\};\,12n-k\rrbracket$.
\end{theorem}

\begin{proof}
The cyclic codes $\mathcal{C}_1$ and $\mathcal{C}_2$ are Hermitian LCD codes over $\mathcal{R}$ whose lengths are $n$ by Proposition \ref{prop14}. Also, $\tau(\mathcal{C}_1)$ and $\tau(\mathcal{C}_2)$ are LCD codes with the parameters $[3n,k_1,d_1]$ and $[3n,k_2,d_2]$, respectively, where
\begin{equation*}
   k_1 = {{\overset{3}{\underset{i=1} {{\displaystyle\sum}}}}}\mathrm{deg}(r_i(x)),
\end{equation*}
and
\begin{equation*}
   k_2 = {{\overset{3}{\underset{j=1} {{\displaystyle\sum}}}}}\mathrm{deg}(s_j(x)),
\end{equation*}
and
\begin{equation*}
d_1=\mathrm{min}\{d(\langle r_1(x)\rangle),d(\langle r_2(x)\rangle),d(\langle r_3(x)\rangle)\},
\end{equation*}
and
\begin{equation*}
d_2=\mathrm{min}\{d(\langle s_1(x)\rangle),d(\langle s_2(x)\rangle),d(\langle s_3(x)\rangle)\}.
\end{equation*}
from Corollary \ref{corollary1} and Proposition \ref{prop2}. The remaining part is similar with the Theorem 4.9 in \cite{euclideansum}. Moreover, by part (i) and Proposition \ref{eaqec}, there exists a binary $\llbracket 12n,k,d\geq \mathrm{min}\{3d_1,2d_2\};12n-k\rrbracket$ EAQEC code.
\end{proof}

\begin{example}
$\mathcal{C}_{1}=\langle (1+x),(1+v)(1+x)+v(1+x+x^2))\rangle $ and
$\mathcal{C}_{2}=\langle ((1+x+x^2),(1+v)(1+x+x^2)+v(1+x))\rangle $ are two
cyclic codes of length $9$ over the ring $\mathcal{R}$. In this case, the dimension of $\mathcal{C}$ is $54$ and $\mathrm{min}\{3d_1,2d_2\}=\mathrm{min}\{6,4\}=4$. Thus, there exists a binary $\llbracket 108,54,d\geq 4;54\rrbracket$ EAQEC code.
\end{example}

Using Theorem \ref{last_theorem}, we establish the existence of several new EAQEC codes for $n=3$ and $n=5$, as detailed in Tables \ref{n=3case} and \ref{n=5case}. These codes are not included in the EAQECC section of \texttt{https://www.codetables.de/} \cite{Grassl:codetables}. The parameters of even more new EAQEC codes, also derived from Theorem \ref{last_theorem} for other values of $n$, are tabulated in Table \ref{eaqectable}. All the computations are done using \verb"MAGMA" \cite{bosma1997magma}.

\begin{table}[!ht]
\caption{Some binary EAQEC Codes from Theorem \ref{last_theorem} in case of $n=3$}\label{n=3case}
 \begin{tabular}{@{}ll@{}}
 \toprule
 $r_i(x),s_i(x),1\leq i\leq 3$ & New EAQEC codes\\
 \midrule
 All of the polynomials are $1+x$. &  $\llbracket 36,12,d\geq 4;24\rrbracket_2$ \\
 One of $r_i(x)$ is $1+x$. The others are $1+x+x^2$. & $\llbracket 36,7,d\geq 4;29\rrbracket_2$ \\
 Two of $r_i(x)$ are $1+x$. The others are $1+x+x^2$. & $\llbracket 36,8,d\geq 4;28\rrbracket_2$ \\
 Three of $r_i(x)$ are $1+x$. The others are $1+x+x^2$. & $\llbracket 36,9,d\geq 4;27\rrbracket_2$ \\
 Four of $r_i(x)$ are $1+x$.The others are $1+x+x^2$. & $\llbracket 36,10,d\geq 4;26\rrbracket_2$ \\
 Five of $r_i(x)$ are $1+x$. The other is $1+x+x^2$. & $\llbracket 36,11,d\geq 4;25\rrbracket_2$ \\
  All of the polynomials are $1+x+x^2$. & $\llbracket 36,6,d\geq 6;30\rrbracket_2$ \\
 \botrule
 \end{tabular}
\end{table}

\begin{table}[!ht]
\caption{Some binary EAQEC Codes from Theorem \ref{last_theorem} in case of $n=5$}\label{n=5case}
 \begin{tabular}{@{}ll@{}}
 \toprule
 $r_i(x),s_i(x),1\leq i\leq 3$ & New EAQEC codes\\
 \midrule
 All of the polynomials are $1+x$. &  $\llbracket 60,24,d\geq 4;36\rrbracket_2$ \\
 One of $r_i(x)$ is $1+x$. The others are $1+x+x^2+x^3+x^4$. & $\llbracket 60,9,d\geq 2;51\rrbracket_2$ \\
 Two of $r_i(x)$ are $1+x$. The others are $1+x+x^2+x^3+x^4$. & $\llbracket 60,12,d\geq 4;48\rrbracket_2$ \\
 Three of $r_i(x)$ are $1+x$. The others are $1+x+x^2+x^3+x^4$. & $\llbracket 60,15,d\geq 4;45\rrbracket_2$ \\
 Four of $r_i(x)$ are $1+x$.The others are $1+x+x^2+x^3+x^4$. & $\llbracket 60,18,d\geq 4;42\rrbracket_2$ \\
 Five of $r_i(x)$ are $1+x$. The other is $1+x+x^2+x^3+x^4$. & $\llbracket 60,21,d\geq 4;39\rrbracket_2$ \\
  All of the polynomials are $1+x+x^2+x^3+x^4$. & $\llbracket 60,6,d\geq 10;54\rrbracket_2$ \\
 \botrule
 \end{tabular}
\end{table}

\begin{table}[!ht]
\caption{Some binary EAQEC Codes from Theorem \ref{last_theorem}}\label{eaqectable}
 \begin{tabular}{@{}ll@{}}
 \toprule
 $r_i(x),s_i(x),1\leq i\leq 3$ & New EAQEC codes\\
 \midrule
 $r_1(x)=r_2(x)=s_3(x)=1+x^2+x^3$, $s_1(x)=s_2(x)=r_3(x)=1+x$. &  $\llbracket 84,30,d\geq 4;64\rrbracket_2$ \\
 $r_1(x)=r_2(x)=s_3(x)=1+x$, $s_1(x)=s_2(x)=r_3(x)=1+x+x^2$. &  $\llbracket 108,54,d\geq 4;54\rrbracket_2$ \\
 $r_1(x)=s_1(x)=s_3(x)=1+x$, $s_2(x)=r_3(x)=1+x+x^2$, $r_2(x)=1+x+x^4$. &  $\llbracket 204,91,d\geq 10;113\rrbracket_2$ \\
 $r_1(x)=s_1(x)=s_3(x)=1+x+x^3$, $s_2(x)=r_3(x)=1+x$, $r_2(x)=1+x^2+x^3$. &  $\llbracket 252,112,d\geq 4;140\rrbracket_2$ \\

 \botrule
 \end{tabular}
\end{table}

\section{Conclusion} \label{sec6}
In this study, we first obtain QEC codes from the Hermitian sums and hulls of cyclic codes over the ring $\mathbb{F}_2\times(\mathbb{F}_2+v\mathbb{F}_2)$, via the Gray map $\tau$. Next, we examine the relationship between the Hermitian hulls of cyclic codes $\mathcal{C}$ over $\mathbb{F}_2\times(\mathbb{F}_2+v\mathbb{F}_2)$ and the Euclidean hulls of their Gray images
$\tau(\mathcal{C})$, which are linear codes of length $3n$ over $\mathbb{F}_2$. This relationship is then used to get binary QEC codes via the Hermitian dual-based quantum construction $X$. Using Hermitian sums and hulls, we obtained new QEC code parameters according to the Euclidean case. Finally, we propose a method to construct EAQEC codes by utilizing matrix product codes of LCD codes over the ring $\mathbb{F}_2\times(\mathbb{F}_2+v\mathbb{F}_2)$. We believe that the theory presented in this work could be generalized for primes $p>2$, which might yield new and good parameters for both QECCs and EAQECCs.

%\bmhead{Acknowledgements}

%The authors would like to thank the editors and the referees for their constructive comments, which have greatly contributed to improving the paper.

\section*{Declarations}

All authors declare that they have no conflict of interest. 


\bibliography{sn-bibliography}% common bib file

@article{F2ring,
  title={Linear codes over $\mathbb{F}_{2}\times (\mathbb{F}_{2}+v\mathbb{F}_{2})$ and the MacWilliams identities},
  author={{\c{C}}al{\i}{\c{s}}kan, Fatma and Aksoy, Refia},
  journal={Applicable Algebra in Engineering, Communication and Computing},
  volume={31},
  number={2},
  pages={135--147},
  year={2020},
  publisher={Springer}
}

@article{cyclicF2,
  title={Cyclic codes over $\mathbb{F}_{2}\times(\mathbb{F}_{2}+v\mathbb{F}_{2})$ and binary quantum codes},
  author={{\c{C}}al{\i}{\c{s}}kan, Fatma and Aksoy, Refia},
  journal={Filomat},
  volume={37},
  number={15},
  pages={5137--5147},
  year={2023}
}

@article{cyclicinring,
  title={Some Results on Cyclic Codes Over $\mathbb{F}_{2}\times(\mathbb{F}_{2}+v\mathbb{F}_{2})$},
  author={Zhu, Shixin and Wang, Yu and Shi, Minjia},
  journal={IEEE Transactions on Information Theory},
  volume={56},
  number={4},
  pages={1680--1684},
  year={2010},
  publisher={IEEE}
}

@inproceedings{quantumxhermitian,
  title={Extending Construction X for quantum error-correcting codes},
  author={Degwekar, Akshay and Guenda, Kenza and Gulliver, T Aaron},
  booktitle={Coding Theory and Applications: 4th International Castle Meeting, Palmela Castle, Portugal, September 15-18, 2014},
  pages={141--152},
  year={2015},
  organization={Springer}
}

@article{euclideansum,
  title={New QEC and EAQEC codes from Euclidean sums and hulls of cyclic codes over $\mathbb{F}_{2}\times(\mathbb{F}_{2}+v\mathbb{F}_{2})$},
  author={Hu, Peng and Liu, Xiusheng},
  journal={Quantum Information Processing},
  volume={24},
  number={6},
  pages={179},
  year={2025},
  publisher={Springer}
}

@article{selfdualcodeinring,
  title={Self-dual codes over $\mathbb{F}_{2}\times(\mathbb{F}_{2}+v\mathbb{F}_{2})$},
  author={Aksoy, Refia and {\c{C}}al{\i}{\c{s}}kan, Fatma},
  journal={Cryptography and Communications},
  volume={13},
  number={1},
  pages={129--141},
  year={2021},
  publisher={Springer}
}

@article{eaqecinring,
  title={Entanglement-assisted binary quantum codes from skew cyclic codes over $\mathbb{F}_{2}\times(\mathbb{F}_{2}+v\mathbb{F}_{2})$},
  author={{\c{C}}al{\i}{\c{s}}kan, Fatma and Aksoy, Refia and Aydin, Nuh and Liu, Peihan},
  journal={Quantum Information Processing},
  volume={22},
  number={5},
  pages={208},
  year={2023},
  publisher={Springer}
}

@article{lcd,
  title={Euclidean and Hermitian LCD MDS codes},
  author={Carlet, Claude and Mesnager, Sihem and Tang, Chunming and Qi, Yanfeng},
  journal={Designs, Codes and Cryptography},
  volume={86},
  number={11},
  pages={2605--2618},
  year={2018},
  publisher={Springer}
}

@article{eaqec,
  title={Linear $\ell$-intersection pairs of codes and their applications},
  author={Guenda, Kenza and Gulliver, T Aaron and Jitman, Somphong and Thipworawimon, Satanan},
  journal={Designs, Codes and Cryptography},
  volume={88},
  number={1},
  pages={133--152},
  year={2020},
  publisher={Springer}
}

@article{benioff1980computer,
  title={The computer as a physical system: A microscopic quantum mechanical Hamiltonian model of computers as represented by Turing machines},
  author={Benioff, Paul},
  journal={Journal of statistical physics},
  volume={22},
  number={5},
  pages={563--591},
  year={1980},
  publisher={Springer}
}

@article{deutsch1985quantum,
  title={Quantum theory, the Church-Turing principle and the universal quantum computer},
  author={Deutsch, David},
  journal={Proceedings of the Royal Society of London. A. Mathematical and Physical Sciences},
  volume={400},
  number={1818},
  pages={97--117},
  year={1985},
  publisher={The Royal Society London}
}

@incollection{feynman2018simulating,
  title={Simulating physics with computers},
  author={Feynman, Richard P},
  booktitle={Feynman and computation},
  pages={133--153},
  year={2018},
  address={New York},
  publisher={CRC Press}
}

@article{simon1997power,
  title={On the power of quantum computation},
  author={Simon, Daniel R},
  journal={SIAM journal on computing},
  volume={26},
  number={5},
  pages={1474--1483},
  year={1997},
  publisher={SIAM}
}

@inproceedings{yao1993quantum,
  title={Quantum circuit complexity},
  author={Yao, A Chi Chih},
  booktitle={Proceedings of 1993 IEEE 34th Annual Foundations of Computer Science},
  pages={352--361},
  year={1993},
  organization={IEEE}
}

@article{hsieh2007general,
  title={General entanglement-assisted quantum error-correcting codes},
  author={Hsieh, Min-Hsiu and Devetak, Igor and Brun, Todd},
  journal={Physical Review A-Atomic, Molecular, and Optical Physics},
  volume={76},
  number={6},
  pages={062313},
  year={2007},
  publisher={APS}
}

@article{calderbank1997quantum,
  title={Quantum error correction and orthogonal geometry},
  author={Calderbank, A Robert and Rains, Eric M and Shor, Peter W and Sloane, Neil JA},
  journal={Physical Review Letters},
  volume={78},
  number={3},
  pages={405},
  year={1997},
  publisher={APS}
}

@article{steane1996multiple,
  title={Multiple-particle interference and quantum error correction},
  author={Steane, Andrew},
  journal={Proceedings of the Royal Society of London. Series A: Mathematical, Physical and Engineering Sciences},
  volume={452},
  number={1954},
  pages={2551--2577},
  year={1996},
  publisher={The Royal Society London}
}

@article{cao2025entanglement,
  title={Entanglement-assisted quantum error-correcting codes using matrix-product codes},
  author={Cao, Meng and Wei, Fuchuan and Luo, Gaojun},
  journal={Designs, Codes and Cryptography},
  pages={1--28},
  year={2025},
  publisher={Springer}
}

@article{li2025eaqec,
  title={EAQEC codes from the LCD codes decomposition of linear codes},
  author={Li, Hui and Liu, Xiusheng},
  journal={Quantum Information Processing},
  volume={24},
  number={2},
  pages={36},
  year={2025},
  publisher={Springer}
}

@article{koroglu2019new,
  title={New entanglement-assisted MDS quantum codes from constacyclic codes},
  author={Koroglu, Mehmet E},
  journal={Quantum Information Processing},
  volume={18},
  number={2},
  pages={28},
  year={2019},
  publisher={Springer}
}

@article{chen2018entanglement,
  title={Entanglement-assisted quantum MDS codes constructed from constacyclic codes},
  author={Chen, Xiaojing and Zhu, Shixin and Kai, Xiaoshan},
  journal={Quantum Information Processing},
  volume={17},
  number={10},
  pages={18},
  year={2018},
  publisher={Springer}
}

@article{pang2021new,
  title={New entanglement-assisted quantum MDS codes},
  author={Pang, Binbin and Zhu, Shixin and Wang, Liqi},
  journal={International Journal of Quantum Information},
  volume={19},
  number={03},
  pages={2150016},
  year={2021},
  publisher={World Scientific}
}

@article{sari2021new,
  title={New entanglement-assisted quantum MDS codes with maximal entanglement},
  author={Sar{\i}, Mustafa and K{\"o}ro{\u{g}}lu, Mehmet E},
  journal={International Journal of Theoretical Physics},
  volume={60},
  number={1},
  pages={243--253},
  year={2021},
  publisher={Springer}
}

@article{li2019entanglement,
  title={Entanglement-assisted quantum MDS codes from generalized Reed--Solomon codes: L. Li et al.},
  author={Li, Lanqiang and Zhu, Shixin and Liu, Li and Kai, Xiaoshan},
  journal={Quantum Information Processing},
  volume={18},
  number={5},
  pages={153},
  year={2019},
  publisher={Springer}
}

@article{pereira2022entanglement,
  title={Entanglement-assisted quantum codes from cyclic codes},
  author={Pereira, Francisco Revson F and Mancini, Stefano},
  journal={Entropy},
  volume={25},
  number={1},
  pages={37},
  year={2022},
  publisher={MDPI}
}

@article{sok2022linear,
  title={Linear codes with arbitrary dimensional hull and their applications to EAQECCs},
  author={Sok, Lin and Qian, Gang},
  journal={Quantum Information Processing},
  volume={21},
  number={2},
  pages={72},
  year={2022},
  publisher={Springer}
}

@article{kai2012new,
  title={New quantum MDS codes from negacyclic codes},
  author={Kai, Xiaoshan and Zhu, Shixin},
  journal={IEEE Transactions on Information Theory},
  volume={59},
  number={2},
  pages={1193--1197},
  year={2012},
  publisher={IEEE}
}

@article{grassl1999quantum,
  title={Quantum BCH codes},
  author={Grassl, Markus and Beth, Thomas},
  journal={arXiv preprint quant-ph/9910060},
  year={1999}
}

@article{ashraf2021new,
  title={New quantum and LCD codes over the finite field of odd characteristic},
  author={Ashraf, Mohammad and Khan, Naim and Mohammad, Ghulam},
  journal={International Journal of Theoretical Physics},
  volume={60},
  number={6},
  pages={2322--2332},
  year={2021},
  publisher={Springer}
}

@article{pereira2021entanglement,
  title={Entanglement-assisted quantum codes from algebraic geometry codes},
  author={Pereira, Francisco Revson F and Pellikaan, Ruud and La Guardia, Giuliano Gadioli and De Assis, Francisco Marcos},
  journal={IEEE Transactions on Information Theory},
  volume={67},
  number={11},
  pages={7110--7120},
  year={2021},
  publisher={IEEE}
}

@Misc{Grassl:codetables,
  author =       "Grassl, Markus",
  title =        "{Bounds on the minimum distance of linear codes and quantum codes}",
  howpublished = "Online available at \url{http://www.codetables.de}",
  year =         "2007",
  note =         "Accessed on 2025-12-10"
}

@article{bosma1997magma,
  title={The Magma algebra system I: The user language},
  author={Bosma, Wieb and Cannon, John and Playoust, Catherine},
  journal={Journal of Symbolic Computation},
  volume={24},
  number={3-4},
  pages={235--265},
  year={1997},
  publisher={Elsevier}
}

\end{document}